\RequirePackage[T1]{fontenc}
\documentclass[conference,letterpaper]{IEEEtran}
\IEEEoverridecommandlockouts
\usepackage{amsmath}
\usepackage{newtxtext,newtxmath}
\usepackage{graphicx}
\usepackage{booktabs}
\usepackage{cite}
\usepackage[hidelinks]{hyperref}
\newcommand{\uM}{\ensuremath{\mu\mathrm{M}}}

\newcommand{\avg}[1]{\left\langle #1\right\rangle}
\newcommand{\dd}{\mathrm{d}}
\newtheorem{proposition}{Proposition}
\graphicspath{{figures/}}
\hypersetup{pdftitle={Reading While Writing: Baseline Information Requirements for Molecular Neural Interfaces},pdfauthor={Hongbin Ni and Ozgur B. Akan}}

\begin{document}
\title{Reading While Writing: Baseline Information\\Requirements for Molecular Neural Interfaces}
\author{\IEEEauthorblockN{Hongbin Ni, \IEEEmembership{Graduate Student Member, IEEE,} and Ozgur B. Akan, \IEEEmembership{Fellow, IEEE}}
\thanks{The authors are with the Centre for neXt Communications (CXC), Electrical Engineering Division, Department of Engineering, University of Cambridge, CB3~0FA Cambridge, U.K. (e-mail: \href{mailto:hn345@cam.ac.uk}{hn345@cam.ac.uk}).\par
Ozgur~B.~Akan is also with the Center for neXt Communications (CXC), Department of Electrical and Electronics Engineering, Ko\c{c} University, 34450 Istanbul, T\"{u}rkiye (e-mail: \href{mailto:oba21@cam.ac.uk}{oba21@cam.ac.uk}, \href{mailto:akan@ku.edu.tr}{akan@ku.edu.tr}).}}

\maketitle

\begin{abstract}
Molecular neural interfaces that deliver and sense the same neurotransmitter must estimate endogenous release in the presence of their own chemical input. We study two exchanging regions with shared saturable uptake and concentration-change measurements. Known delivery and exchange permit identification of apparent uptake at the driven region, but an unknown resting concentration difference leaves the source region's incremental uptake uncertain. We characterize all admissible source histories consistent with both ideal records and show how delivery and baseline information affect release classification. A dopamine-inspired example produces identical ideal concentration-change records for mean release-rate changes of $-8.99$ and $+1.31$ nanomolar per second. Stronger delivery changes whether all compatible sources imply suppression, with the true source held fixed. Under a Gaussian observation model, we bound discrimination based on chemical fluctuations. A baseline reference with an assumed standard deviation of 3.3 nanomolar reduces root-mean-square error from 10.41 to 1.21 nanomolar per second compared with an incorrect fixed baseline under exact gain calibration. Classification errors remain substantial with weak excitation, gain errors or responses near the category boundaries.
\end{abstract}


\section{Introduction}
Chemical delivery and neurochemical sensing connect biological and digital systems in proposed dopamine read/write interfaces~\cite{Baradari2025}. Same-species two-way molecular communication motivates treatment of self-interference~\cite{Kwak2020}, while receiver saturation limits operation~\cite{Zabini2023}. ART-Rx adapts decision thresholds for molecular symbol detection through feedback control~\cite{NiAkan2025}. Endogenous release estimation requires accounting for uptake of delivered and endogenous molecules, even with a calibrated linear sensor.

Resting dopamine affects apparent uptake~\cite{ChenBudygin2007,Dreyer2010}, and basal turnover affects clearance after exogenous delivery~\cite{Michael2005BasalTurnover}. Nonlinear compartment transformations and target uncertainty among indistinguishable models are established~\cite{Chapman1994,Wu2016}. Prior work examines conclusions shared by nonidentifiable models and disagreement among their predictions~\cite{Cedersund2012,Letham2016}. Functional identifiability studies analyze ambiguities involving unknown dynamics, parameters, and unobserved initial states~\cite{LomanBaker2025,LomanBrowningBaker2026}. Combined absolute and differential measurements provide additional baseline information~\cite{Atcherley2015}.

We study endogenous release inference from two concentration-change records with known delivery and exchange. Sufficient excitation identifies apparent uptake at the driven region. Transferring this calibration to the source region requires the unobserved resting concentration difference. Spatial coupling and nonlinear uptake jointly create this ambiguity.

For the stated model, we characterize the complete family of nonnegative sources consistent with both ideal records and its full range of mean release-rate changes. Changing delivery can alter whether every compatible source supports suppression while the generating source remains fixed. We establish baseline and repeated-source conditions that resolve the ambiguity, and assess discrimination from noisy records and reconstruction sensitivity to calibration error.

\section{System Model and Inference Task}
\label{sec:model}
As shown in Fig.~\ref{fig:pair}, we represent local chemical delivery and neighboring endogenous release by two coupled, well-mixed extracellular regions. Both regions exhibit Michaelis--Menten uptake and exchange a single molecular species. For equal effective volumes, their concentrations $c_i(t)$ obey
\begin{align}
\dot c_1&=q+d(c_2-c_1)-U(c_1),\label{eq:plant1}\\
\dot c_2&=r+d(c_1-c_2)-U(c_2),\quad
U(c)=\frac{Vc}{K+c}.\label{eq:plant2}
\end{align}
Concentrations are in \uM, rates $q,r,V$ per effective volume in \uM/s, and $d$ in $\mathrm{s}^{-1}$. The receiver knows delivery $q(t)$ and exchange rate $d>0$. Intrinsic uptake parameters $V,K>0$ are shared and unknown. Region 1 has uptake but no endogenous release, an allocation considered in spatial dopamine models~\cite{Venton2003}. Endogenous release in region 2 has an unknown continuous nonnegative rate $r(t)$.

Before delivery, $q=0$ and the system is at a positive equilibrium $b=(b_1,b_2)$, giving
\begin{equation}
b_2=b_1+\frac{U(b_1)}d,\qquad r_0=U(b_1)+U(b_2).\label{eq:rest}
\end{equation}
The ideal differential records are $x_i=c_i-b_i$ in calibrated concentration units. A known filter $\tau\dot z_i=x_i-z_i$, initialized at rest, preserves this information in continuous noiseless records because $x_i=z_i+\tau\dot z_i$. Inversion of sampled noisy records is addressed separately. Subtracting a prestimulation reference defines zero concentration change, an absolute baseline measurement supplies $b_i$, and gain calibration supplies the concentration scale.

For a window $[t_1,t_2]$ of length $L$, the target is the mean change in release rate from its resting value $r_0$:
\begin{equation}
T=\frac1L\int_{t_1}^{t_2}[r(t)-r_0]\,\dd t.\label{eq:target}
\end{equation}
The target window is $[t_1,t_2]=[20,40]$ s. Reconstruction uses the completed 45 s observation record, excluding computational latency. Suppression means $T<-\delta$, equivalence $|T|\leq\delta$, and enhancement $T>\delta$, with illustrative $\delta=2.5$ nM/s. Each inference uses one intervention record. The delivery comparison holds the generating source fixed to isolate an inference effect. Source repeatability across trials is additional information, introduced separately in Section~\ref{sec:repeated}.

\begin{figure}[t]
\centering
\includegraphics[width=\columnwidth]{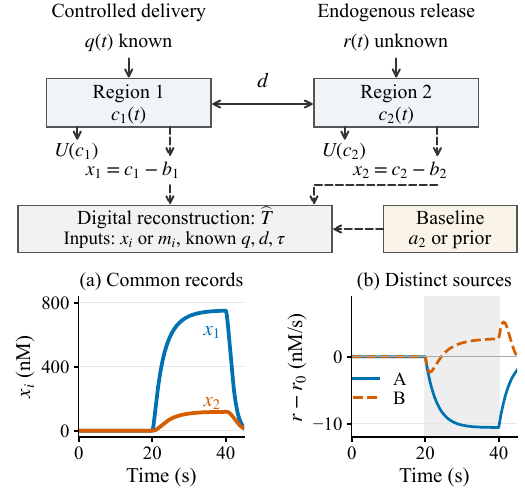}
\caption{Effective extracellular regions and digital readout. Solid arrows show chemical processes and dashed arrows show information paths. Noisy $m_i$ are defined in Section~\ref{sec:receiver}. At $q_{\max}=0.40$ \uM/s, model A uses Table~\ref{tab:parameters} and model B uses $\beta=5$ nM in~\eqref{eq:family}--\eqref{eq:source}. (a) Ideal records coincide. (b) Release-rate changes average $-8.99$ and $+1.31$ nM/s over the shaded window.}
\label{fig:pair}
\end{figure}

\section{Compatible Sources and Baseline Information}
\label{sec:family}
\subsection{Complete Compatible Family}
Define the apparent parameters and their invariant product
\begin{equation}
A=\frac{VK}{K+b_1},\qquad B=K+b_1,\qquad P=AB=VK.\label{eq:apparent}
\end{equation}
Subtracting the first resting balance from~\eqref{eq:plant1} yields the observed quantity
\begin{equation}
w=q+d(x_2-x_1)-\dot x_1=\frac{Ax_1}{B+x_1}.\label{eq:increment}
\end{equation}
This apparent-uptake form follows from the established baseline transformation~\cite{ChenBudygin2007,Dreyer2010}.

\begin{proposition}\label{prop:complete}
Under the assumptions in Section~\ref{sec:model}, suppose $x_1$ takes two distinct nonzero values. The ideal records uniquely determine $A,B$. Every compatible model is indexed by $\beta=b_1$ through
\begin{align}
K_\beta&=B-\beta,& V_\beta&=\frac{P}{B-\beta},\nonumber\\
b_{2,\beta}&=\beta+\frac{A\beta}{d(B-\beta)},&
c_{i,\beta}&=b_{i,\beta}+x_i,\label{eq:family}
\end{align}
where $b_{1,\beta}=\beta$, $0<\beta<B$, and
\begin{equation}
r_\beta=\dot x_2-d(x_1-x_2)+d(b_{2,\beta}-\beta)
+U_\beta(b_{2,\beta}+x_2).\label{eq:source}
\end{equation}
The admissible family consists exactly of members satisfying the parameter restrictions, $c_{i,\beta}>0$, and $r_\beta\geq0$ throughout the observation horizon.
\end{proposition}
\begin{IEEEproof}
For nonzero $x_1$,~\eqref{eq:increment} gives $x_1/w=B/A+x_1/A$. Two distinct values fix its slope and intercept, hence $A,B$. Setting $b_1=\beta$ then fixes $K_\beta,V_\beta$ by~\eqref{eq:apparent}, $b_{2,\beta}$ by~\eqref{eq:rest}, and $r_\beta$ by the second balance. This proves necessity. Conversely,~\eqref{eq:family} preserves $A,B$ and the first incremental balance. Equation~\eqref{eq:source} preserves the second balance and its resting value, including the constant exchange term. The states reproduce $x_i$ and have locally Lipschitz dynamics on the positive concentration domain. Every feasible member is therefore a compatible solution.
\end{IEEEproof}

\subsection{Release Target and Additional Information}
The resting concentration difference enters the source region's uptake law. Write $\Delta b=b_2-b_1$ and
\begin{equation}
C_\beta=K_\beta+b_{2,\beta}
=B+\frac{A\beta}{d(B-\beta)}.\label{eq:C}
\end{equation}
The second region's incremental uptake is $P x_2/[C_\beta(C_\beta+x_2)]$. Thus calibrating $A,B$ in region 1 leaves the second-region law dependent on the unknown resting difference $\Delta b$. Integrating its mass balance expresses the target in terms of the records and each candidate baseline:
\begin{align}
T_\beta&=\frac{x_2(t_2)-x_2(t_1)}L-d\avg{x_1-x_2}\nonumber\\
&\quad+P\avg{\frac{x_2}{C_\beta(C_\beta+x_2)}}.\label{eq:functional}
\end{align}
Here $\avg{\cdot}$ denotes the window average. For a compact feasible interval $\mathcal B$, continuity gives the complete target interval $\mathcal T=[\min_{\mathcal B}T_\beta,\max_{\mathcal B}T_\beta]$. Every compatible source supports suppression precisely when this interval's upper endpoint is below $-\delta$.

The sensitivity to the hidden baseline follows directly:
\begin{equation}
\frac{\partial T_\beta}{\partial\beta}
=P C'_\beta\avg{\frac{1}{(C_\beta+x_2)^2}-\frac{1}{C_\beta^2}}.\label{eq:target_sensitivity}
\end{equation}
Since $C'_\beta>0$, a nonnegative $x_2$ makes $T_\beta$ nonincreasing in $\beta$, while a nonpositive $x_2$ reverses that ordering. For mixed-sign records, the derivative's sign depends on the window average in~\eqref{eq:target_sensitivity}. Thus baseline uncertainty affects the release target through the observed trajectory, linking the calibration requirement to the operating condition.

Under the stated model, including shared uptake and known exchange, an exact baseline at either region determines $\beta$. For example, a second-region value $h$ gives the unique physical root
\begin{equation}
\beta(h)=\frac{2Bh}{S+\sqrt{S^2-4Bh}},\qquad S=B+A/d+h.\label{eq:inverse}
\end{equation}
A known resting concentration difference gives $\beta=Bd\Delta b/(A+d\Delta b)$. Known intrinsic $K$ or $V$ also fixes $\beta$. A bounded baseline observation intersects $\mathcal B$ with its corresponding inverse image and contracts $\mathcal T$.

A single region with calibrated incremental uptake $F(x)$ gives $r-r_0=\dot x-q+F(x)$ directly. Shared linear uptake also eliminates baseline dependence from both incremental laws. Thus spatial coupling and nonlinear uptake jointly create the present ambiguity.

\subsection{Repeated-Source Calibration as an Alternative}
\label{sec:repeated}
Two deliveries with the same resting concentrations, uptake parameters and unknown source $r(t)$ provide an alternative baseline calibration. For ordered inputs $q_h\geq q_l$ that differ on a nonzero time interval within the observed horizon, comparison of the cooperative dynamics gives $x_{2,h}>x_{2,l}$ after their effects propagate. For an incorrect candidate baseline, define
\begin{equation}
H_\beta(x)=\frac{Px}{C_\beta(C_\beta+x)}
-\frac{Px}{C_{b_1}(C_{b_1}+x)}.\label{eq:repeatH}
\end{equation}
Its derivative has a constant nonzero sign when $\beta\ne b_1$. The candidate sources satisfy $r_{\beta,h}-r_{\beta,l}=H_\beta(x_{2,h})-H_\beta(x_{2,l})\ne0$. Requiring the same source across both records therefore excludes every incorrect candidate. Source repeatability is an additional assumption, independent of the unknown baseline.

\section{Finite-Record Discrimination}
\label{sec:noise}
Distinct absolute concentrations can produce distinct chemical fluctuations even when $x_i$ coincide. We quantify this information using the linear-noise approximation (LNA)~\cite{Komorowski2011}. Along a deterministic trajectory, concentration fluctuations obey $\dd\xi=J\xi\,\dd t+Q_c^{1/2}\dd W$, with
\begin{align}
J&=\begin{bmatrix}-d-U'(c_1)&d\\d&-d-U'(c_2)\end{bmatrix},\nonumber\\
Q_c&=\Omega^{-1}\!\biggl[\operatorname{diag}(q+U(c_1),r+U(c_2))\nonumber\\
&\hspace{4em}+d(c_1+c_2)\begin{bmatrix}1&-1\\-1&1\end{bmatrix}\biggr].\label{eq:lna}
\end{align}
This law treats source, uptake and exchange events as independent Poisson reactions. Known delivery specifies the intensity $\Omega q(t)$, while individual injection events remain unobserved. Here $\Omega$ is the number of molecules per \uM. Filtered fluctuations satisfy $\tau\dot\nu=\xi-\nu$. For the augmented state $(\xi^\mathsf T,\nu^\mathsf T)^\mathsf T$, define
\begin{equation}
F=\begin{bmatrix}J&0\\I/\tau&-I/\tau\end{bmatrix},\quad
Q=\begin{bmatrix}Q_c&0\\0&0\end{bmatrix},\quad S_z=[0\ I].\label{eq:augmented}
\end{equation}
The initial state is zero-mean Gaussian with stationary resting covariance $\Pi_0$, where $F_0\Pi_0+\Pi_0F_0^\mathsf T+Q_0=0$. Thereafter $\dot\Pi=F\Pi+\Pi F^\mathsf T+Q$. If $\Phi(t,s)$ is the state-transition matrix generated by $F$, the chemical observation covariance is
\begin{equation}
K_\nu(t,s)=S_z\Phi(t,s)\Pi(s)S_z^\mathsf T,\qquad t\geq s.\label{eq:covariance}
\end{equation}

Without an absolute baseline measurement, the sampled record is $Y_i[k]=z_i(t_k)+\nu_i(t_k)+\epsilon_i[k]$ at $t_k=0.1k$ s, $k=0,\ldots,450$. Electronic noise is jointly zero-mean Gaussian, independent of chemical fluctuations and across channels, with $\operatorname{Cov}(\epsilon_i[k],\epsilon_i[l])=\sigma_e^2\rho^{|k-l|}$. Adding this covariance to~\eqref{eq:covariance} gives the covariance of the stacked record. Two compatible models therefore give Gaussian distributions with equal means and potentially different $902\times902$ covariance matrices $\Sigma_a,\Sigma_b$. Their Kullback--Leibler divergence, using natural logarithms, is
\begin{equation}
D_{a\|b}=\frac12\left[\operatorname{tr}(\Sigma_b^{-1}\Sigma_a)-902
+\log\frac{\det\Sigma_b}{\det\Sigma_a}\right].\label{eq:kl}
\end{equation}
Let $\varepsilon=\sqrt{\min(D_{a\|b},D_{b\|a})/2}$. Pinsker's inequality~\cite{CoverThomas2006} bounds total variation by $\varepsilon$. For a suppression decision event $\mathcal D$, this gives
\begin{equation}
\Pr_a(\mathcal D)\leq\Pr_b(\mathcal D)+\varepsilon.\label{eq:power}
\end{equation}
It also gives an equal-prior binary error probability of at least $(1-\varepsilon)/2$. Applying the same temporal referencing to both models cannot increase distinguishability. This bound applies to the stated Gaussian observation law and finite record. Adding a baseline measurement changes the available information.

\section{Numerical Results and Discussion}
\label{sec:results}

\begin{figure*}[t]
\centering
\includegraphics[width=\textwidth]{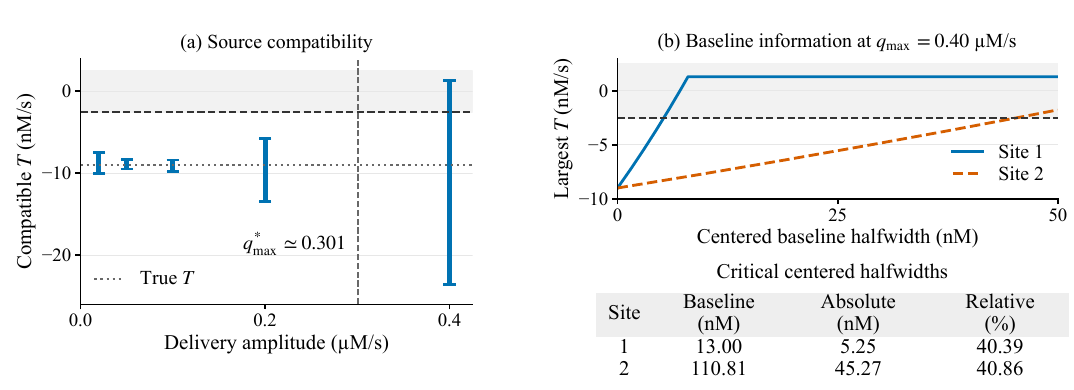}
\caption{Actuation and baseline information with ideal differential records. (a) Complete compatible target ranges at five deliveries, with true mean release-rate change $-8.99$ nM/s held fixed. (b) Largest compatible target after a centered baseline constraint at either region, at $q_{\max}=0.40$ \uM/s. Gray shading denotes equivalence and the dashed line marks $-\delta$. The tabulated critical halfwidths compare absolute and proportional precision. Suppression for every compatible source requires a smaller halfwidth.}
\label{fig:actuation}
\end{figure*}
\subsection{Operating Example and Continuous Compatibility}
The dopamine-inspired example in Table~\ref{tab:parameters} uses selected uptake scales motivated by regional modeling~\cite{Venton2003} and a baseline scale motivated by microdialysis~\cite{Hershey2013}. Values from different preparations define a composite illustrative scenario. Exchange, effective volume and observation dynamics are chosen modeling parameters. Prestimulation fast-scan controlled-adsorption voltammetry can supply tonic dopamine concentration before differential fast-scan cyclic voltammetry recording~\cite{Atcherley2015}. This motivates the reference principle. Its acquisition and settling precede our 45 s record, and its assumed precision is specified independently.

\begin{table}[t]
\caption{Model and Observation Parameters. SD denotes standard deviation.}
\label{tab:parameters}
\centering
\begin{tabular}{@{}lll@{}}
\toprule
Quantity & Value & Status\\
\midrule
$V,K$ & $0.58$ \uM/s, $0.58$ \uM & Illustrative scales~\cite{Venton2003}\\
$b_1,b_2$ & $13$, $110.81$ nM & Scale~\cite{Hershey2013}, equilibrium\\
$d$, volume/region & $0.13$ s$^{-1}$, $10$ pL & Chosen\\
$\tau$, sampling & $0.2$ s, $10$ Hz & Chosen filter, sampling~\cite{Keithley2011}\\
$\sigma_e,\rho$ & $4.67$ nM, $0.3$ & Chosen electronics\\
$\sigma_a$, gain SD & $3.3$ nM, $0,1,5\%$ & Chosen calibration\\
Observation horizon & $45$ s & Record duration\\
$\delta$ & $2.5$ nM/s & Category criterion\\
\bottomrule
\end{tabular}
\end{table}

Let $p_\tau$ solve $\tau\dot p_\tau+p_\tau=\mathbf1_{[20,40]}$ with zero initial value. The actuation study uses $q=q_{\max}p_{0.5}$ and $r=r_0(1-0.1p_3)$, giving $T=-8.9907$ nM/s. Prior restrictions are $5\leq\beta\leq100$ nM, uptake parameters within $\pm20\%$ of the nominal entries in Table~\ref{tab:parameters}, and $b_{2,\beta}\leq250$ nM. Their static intersection is $\mathcal B=[5,28.6302]$ nM.

This entire interval is dynamically admissible. The cooperative system has the constant lower barrier $\ell_1=0.011$ \uM, $\ell_2=\ell_1+U(\ell_1)/d\simeq0.09404$ \uM. At that point the first drift is $q\geq0$ and the second is at least $0.0034585$ \uM/s. Initial concentrations exceed the barrier. Alternative concentrations therefore exceed $0.003$ and $0.02533$ \uM. Moreover,
\begin{equation}
\frac{\partial r_\beta}{\partial\beta}
=\frac{2P}{(B-\beta)^2}+\frac{P C'_\beta}{(C_\beta+x_2)^2}>0.\label{eq:sourcepositive}
\end{equation}
At $\beta_{\min}$, define $H(x)=f_{C_{\beta_{\min}}}(x)-f_{C_{b_1}}(x)$, where $f_C(x)=Px/[C(C+x)]$. Since $C_{\beta_{\min}}<C_{b_1}$, $H'(x)>0$. Thus $r_{\beta_{\min}}\geq r_{0,\beta_{\min}}-0.1r_0+H(\ell_2-b_2)=0.0299985$ \uM/s, certifying positive sources throughout the interval for every tested nonnegative delivery.

Figure~\ref{fig:actuation} shows the complete target intervals. For $q_{\max}=0.20$ \uM/s, all targets lie in $[-13.48,-5.78]$ nM/s and support suppression. At $0.40$ \uM/s the range is $[-23.64,1.31]$ nM/s, admitting both suppression and equivalence. Figure~\ref{fig:pair} compares the generating baseline $\beta_a=13$ nM with $\beta_b=5$ nM. The latter gives $b_{2,b}=42.11$ nM, $V_b=0.5721$ \uM/s and $K_b=0.588$ \uM. Equation~\eqref{eq:source} defines its source, yielding $T_b=1.3136$ nM/s and coincident differential records.

The extrema are justified over continuous $\beta$ using~\eqref{eq:target_sensitivity}. For a linear concentration segment with endpoints $a,b$, the segment-average derivative with respect to $C$ is $P[(C+a)^{-1}(C+b)^{-1}-C^{-2}]$. Weighting by segment duration divided by $L$, splitting at zero crossings and bounding over the $C$ interval gives a uniform sign for each tested delivery. Thus endpoints attain the extrema. At $q_{\max}=0.10$ \uM/s the sign is negative. Cooperative input ordering increases $x_2$ with delivery and preserves that sign. Hence $\beta_{\min}$ maximizes the target throughout $[0.20,0.40]$ \uM/s. Since $H'(x)>0$ there, the maximum rises strictly with delivery and crosses $-\delta$ once, near $q_{\max}=0.301$ \uM/s. The same scalar ambiguity persists on both sides, while the supported category changes.

Dynamics use relative/absolute tolerances $10^{-10}/10^{-12}$ and maximum step $0.025$ s. Target quadrature uses a $0.00625$ s grid, checked at twice and half this spacing. Piecewise-linear target shifts agree with dense-solution quadrature within $1.61\times10^{-7}$ nM/s, and~\eqref{eq:functional} agrees with the equivalent uptake-difference form within $4.28\times10^{-7}$ nM/s. Concavity bounds enclose segment integrals. A second-derivative interpolation allowance below $0.0072$ nM/s gives the crossing bracket $[0.30096,0.30132]$ \uM/s, conditional on exact trajectory knots. Global ODE error is assessed by refinement, without a validated interval ODE bound.

At delivery $0.40$ \uM/s, consider baseline constraints centered on the generating values, $|b_{i,\beta}-b_i|\leq\epsilon_i$. The category boundary satisfies $T_{\beta_*}=-\delta$ at $\beta_*\simeq7.7499$ nM. Since $T_\beta$ decreases and $b_{i,\beta}$ increases with $\beta$, suppression for every compatible source requires
\begin{equation}
\epsilon_i<b_i-b_{i,\beta_*}.\label{eq:anchorprecision}
\end{equation}
The corresponding halfwidths are $5.25$ nM at region 1 and $45.27$ nM at region 2, approximately $40.4\%$ and $40.9\%$ of their baselines. The apparent location advantage therefore depends on the error budget. In this example, $b_2(\beta)=\beta[1+A/(d(B-\beta))]$ is nearly proportional to $\beta$ over the relevant range. Equal fractional bounds constrain $\beta$ almost equally, explaining the similar proportional tolerances. These are centered deterministic baseline bounds with exact apparent uptake parameters. They differ from random measurement errors and from gain uncertainty affecting the differential records. The crossing depends on the source, prior and category threshold.

\subsection{Information in Fluctuations}
For the pair in Fig.~\ref{fig:pair}, 10 pL gives $\Omega=6.02214\times10^6$ molecules per \uM. The model-specific LNA covariances and electronic noise in Table~\ref{tab:parameters} yield $\min D=3.319\times10^{-4}$ and $\varepsilon=0.01288$. A rule with at most $5\%$ false-suppression probability under model B has at most $6.29\%$ detection probability under model A. The equal-prior binary error probability is at least $49.36\%$. This calculation includes covariance information and remains conditional on the chosen volume, noise and Gaussian approximation.

\subsection{Fitted-Kinetics Receiver and Calibration Sensitivity}
\label{sec:receiver}
We next simulate $T\in\{-9,-3,-2,0,2,3,9\}$ nM/s at $q_{\max}\in\{0.05,0.20,0.40\}$ \uM/s. The source is $r=r_0+\alpha p_3$, with $\alpha=T/\avg{p_3}$ and all sources positive. Each condition has 1,000 independent realizations. Comparisons use paired Gaussian random draws, with chemical covariance evaluated along each condition's trajectory. Gains $g_i=1+\eta_i$ remain constant within each record. The errors $\eta_i\sim\mathcal N(0,s_g^2)$ are independent across channels and independent of chemical and electronic noise. Measurements are $m_i=g_i(z_i+\nu_i)+\epsilon_i$, and the added baseline observation is
\begin{equation}
a_2=g_2(b_2+\Delta_2)+\nu_a,\qquad
\nu_a\sim\mathcal N(0,\sigma_a^2).\label{eq:anchor}
\end{equation}
Here $\Delta_2$ is a fixed transfer offset between the earlier baseline acquisition and the current resting reference, and $\nu_a$ is independent additive noise. Unless stated otherwise, $\Delta_2=0$. The reference and differential channel share $g_2$, while reconstruction assumes $g_i=1$ and estimates $A,B$. The 3.3 nM additive reference-noise SD is assumed. Acquisition overhead and time-varying reference drift require preparation-specific characterization.

The receiver subtracts the mean of samples before 10 s, then applies a seven-sample cubic Savitzky--Golay smoother and derivative to form $\hat x=\hat z+\tau\widehat{\dot z}$. At record edges, the local polynomial is evaluated within the available samples. On the 25 one-second intervals from 20 to 45 s, it fits the integrated first balance by three-pass feasible generalized least squares. For interval $j$, its residual is
\begin{equation}
e_j=\Delta\hat x_{1,j}-\int_jq\,\dd t
-d\int_j(\hat x_2-\hat x_1)\,\dd t
+A\int_j\frac{\hat x_1}{B+\hat x_1}\,\dd t.\label{eq:gls}
\end{equation}
Integrals use composite Simpson quadrature. Let $\mathsf I$ and $\mathsf E$ implement these interval integrals and endpoint differences, and let $\Lambda$ implement referencing and filter inversion. With $R_{kl}=\rho^{|k-l|}$, each processed electronic covariance is $\Gamma=\sigma_e^2\Lambda R\Lambda^\mathsf T$. Linearizing~\eqref{eq:gls} gives
\begin{align}
M_1&=\mathsf E+d\mathsf I+\mathsf I\operatorname{diag}\!\left(\frac{AB}{(B+\hat x_1)^2}\right),\nonumber\\
M_2&=-d\mathsf I,\qquad
W_e=M_1\Gamma M_1^\mathsf T+M_2\Gamma M_2^\mathsf T.\label{eq:weights}
\end{align}
Each pass fixes $W_e$ at the current $A,B$ and minimizes $e^\mathsf T W_e^{-1}e$ over logarithmic parameters. Bound-constrained trust-region reflective optimization uses relative cost, step and scaled-gradient tolerances of $10^{-9}$. Simulated records include chemical noise, while the fitting weights use the electronic-noise approximation above. Initialization is $(A,B)=(0.55,\max(0.6,B_{\min}+0.001))$, with $A\in[0.02,2]$ \uM/s, $B\in[B_{\min},2]$ \uM{} and $B_{\min}=\max(0.05,-\min\hat x_1+0.005)$ \uM. Equations~\eqref{eq:inverse} and~\eqref{eq:functional} then give $\hat T$. Categories use the same fixed $\delta$. Performance is measured by root-mean-square error (RMSE) and category-error frequency.

\begin{figure}[t]
\centering
\includegraphics[width=\columnwidth]{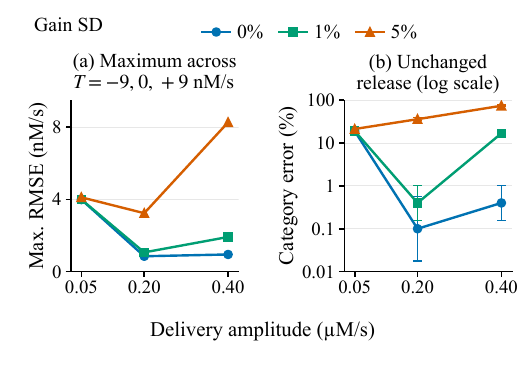}
\caption{Baseline-assisted reconstruction with a 3.3 nM additive reference SD. Maximum RMSE is taken across $T=-9,0,9$ nM/s. Category errors for unchanged release ($T=0$) use a logarithmic scale and 95\% Wilson intervals from 1,000 realizations per condition. Gains are constant within each record, with the indicated SD across realizations. Lines connect the three tested deliveries.}
\label{fig:receiver}
\end{figure}

Figure~\ref{fig:receiver} separates excitation and calibration effects. The $T=+9$ nM/s case gives the maximum RMSE at every displayed setting. At delivery $0.20$ \uM/s and $1\%$ gain SD, this maximum is $1.06$ nM/s. Category error for unchanged release is $0.4\%$, and both $T=\pm9$ nM/s cases have zero observed errors. The two-sided Wilson upper limit for zero errors in 1,000 is $0.383\%$.

Weak delivery limits parameter estimation. At $0.05$ \uM/s and exact gains, category error for unchanged release is $19.0\%$ and $28.4\%$ of fits lie near a parameter bound. Supplying the true $A,B$ reduces this error to $0.2\%$. At $0.40$ \uM/s, error for unchanged release rises from $0.4\%$ at exact gain to $16.9\%$ at $1\%$ gain SD and $74.3\%$ at $5\%$ gain SD, with a 95\% Wilson interval of $[71.5,76.9]\%$ for the latter. Cases near the category boundaries, $T=\pm2,\pm3$ nM/s, have approximately $22$--$33\%$ error at exact gain for the two higher deliveries. All optimizations succeeded and estimates were finite. These classifications use point estimates, with accuracy depending on the response and calibration.

An initialization check used 100 unchanged-release records at each of $q_{\max}=0.05,0.20$ \uM/s with exact gains, sampled uniformly without replacement before inspecting fits and paired across conditions. Besides the original start, we used $A_0\in\{0.05,0.2,0.8,1.5\}$ \uM/s and $B_0=B_{\min}+f(2-B_{\min})$, with $f\in\{0.2,0.8\}$. All starts converged and retained each record's category. Maximum within-record target spreads were below $0.041$ nM/s for three-pass fits and $4.35\times10^{-4}$ nM/s with the original final-pass weights fixed. The checked weak-input failures persisted across these initializations.

Fixed reference offsets also affect reconstruction. At $q_{\max}=0.40$ \uM/s, $T=0$ and exact gains, offsets $\Delta_2=-5$ and $+5$ nM give category errors of $2.5\%$ and $3.2\%$, respectively, compared with $0.4\%$ at zero offset. These paired 1,000-realization comparisons retain $\sigma_a=3.3$ nM and quantify sensitivity to a fixed baseline transfer error.

Table~\ref{tab:alias} evaluates the alternative source in Fig.~\ref{fig:pair}, with $T=1.3136$ nM/s in the equivalence band. Both rows use the same reconstruction and differential records with different baseline inputs. The noisy reference reduces RMSE and false suppression, while total category error remains $15.7\%$, mainly from crossing the nearby enhancement boundary.

\begin{table}[t]
\caption{Reconstruction of the Alternative Source, with True $b_2=42.11$ nM. Delivery is 0.40 \uM/s, gains are exact, and each row uses 1,000 realizations.}
\label{tab:alias}
\centering
\begin{tabular}{@{}lccc@{}}
\toprule
Baseline information & RMSE & False & Total category\\
 & (nM/s) & suppression & error\\
\midrule
Assumed $b_2=110.81$ nM & 10.41 & 100\% & 100\%\\
Noisy reference, SD 3.3 nM & 1.21 & 0.2\% & 15.7\%\\
\bottomrule
\end{tabular}
\end{table}

\subsection{Implications for Interface Operation}
Weak delivery can support suppression in ideal records despite poor excitation for kinetic estimation. Stronger delivery improves excitation but can enlarge the compatible target range. Gain errors perturb the accumulation and exchange terms in~\eqref{eq:functional}, the fitted uptake law, and the baseline reference. Larger delivery-induced concentration changes can therefore increase reconstruction error. The middle tested delivery gives the smallest maximum RMSE across the three specified release changes for each examined gain setting. Delivery and calibration should be evaluated jointly. Conclusions that hold throughout $\mathcal T$ apply to all compatible ideal models, while classifications from noisy records use point estimates.

Preparation-specific application requires calibrated delivery per effective volume, exchange, concentration scale and sensor dynamics, and evidence for the assumed uptake and release allocation. Admissibility ensures nonnegative release within these mass balances. Cellular mechanisms may impose further constraints. Uncertain transport and estimation of sensor gains remain extensions.

\section{Conclusion}
In the specified two-region model, an unknown resting concentration difference can leave endogenous release ambiguous despite identifiable local apparent uptake. We characterize the complete compatible source family and establish when every compatible source supports suppression without unique source recovery. Delivery can change this conclusion with the generating source fixed. Baseline information narrows or resolves the ambiguity, with repeated-source calibration providing an alternative. Under a Gaussian observation model, finite-record analysis bounds discrimination, while simulations demonstrate accurate baseline-assisted reconstruction at selected operating conditions. These findings support joint evaluation of delivery and calibration in molecular neural-interface design.

\bibliographystyle{IEEEtran}
\bibliography{references}

@article{ChenBudygin2007,
  author  = {Chen, Kevin C. and Budygin, Evgeny A.},
  title   = {Extracting the basal extracellular dopamine concentrations from the evoked responses: Re-analysis of the dopamine kinetics},
  journal = {J. Neurosci. Methods},
  year    = {2007},
  month   = aug,
  volume  = {164},
  number  = {1},
  pages   = {27--42},
  doi     = {10.1016/j.jneumeth.2007.03.020},
}

@article{Michael2005BasalTurnover,
  author  = {Michael, Adrian C. and Borland, Laura M. and Mitala, Jr., Joseph J. and Willoughby, Bridget M. and Motzko, Christina M.},
  title   = {Theory for the impact of basal turnover on dopamine clearance kinetics in the rat striatum after medial forebrain bundle stimulation and pressure ejection},
  journal = {J. Neurochem.},
  year    = {2005},
  month   = sep,
  volume  = {94},
  number  = {5},
  pages   = {1202--1211},
  doi     = {10.1111/j.1471-4159.2005.03265.x},
}

@article{Dreyer2010,
  author  = {Dreyer, Jakob K. and Herrik, Kjartan F. and Berg, Rune W. and Hounsgaard, J{\o}rn D.},
  title   = {Influence of Phasic and Tonic Dopamine Release on Receptor Activation},
  journal = {J. Neurosci.},
  year    = {2010},
  month   = oct,
  volume  = {30},
  number  = {42},
  pages   = {14273--14283},
  doi     = {10.1523/JNEUROSCI.1894-10.2010},
}

@article{Chapman1994,
  author  = {Chapman, Michael J. and Godfrey, Keith R. and Vajda, Sandor},
  title   = {Indistinguishability for a class of nonlinear compartmental models},
  journal = {Math. Biosci.},
  year    = {1994},
  month   = jan,
  volume  = {119},
  number  = {1},
  pages   = {77--95},
  doi     = {10.1016/0025-5564(94)90005-1},
}

@article{Wu2016,
  author  = {Wu, Xiaotian and Nekka, Fahima and Li, Jun},
  title   = {Steady-state volume of distribution of two-compartment models with simultaneous linear and saturated elimination},
  journal = {J. Pharmacokinet. Pharmacodyn.},
  year    = {2016},
  month   = aug,
  volume  = {43},
  number  = {4},
  pages   = {447--459},
  doi     = {10.1007/s10928-016-9483-z},
}

@article{Komorowski2011,
  author  = {Komorowski, Micha{\l} and Costa, Maria J. and Rand, David A. and Stumpf, Michael P. H.},
  title   = {Sensitivity, robustness, and identifiability in stochastic chemical kinetics models},
  journal = {Proc. Natl. Acad. Sci. USA},
  year    = {2011},
  month   = may,
  volume  = {108},
  number  = {21},
  pages   = {8645--8650},
  doi     = {10.1073/pnas.1015814108},
}

@article{Atcherley2015,
  author  = {Atcherley, Christopher W. and Wood, Kevin M. and Parent, Kate L. and Hashemi, Parastoo and Heien, Michael L.},
  title   = {The coaction of tonic and phasic dopamine dynamics},
  journal = {Chem. Commun.},
  year    = {2015},
  month   = feb,
  volume  = {51},
  number  = {12},
  pages   = {2235--2238},
  doi     = {10.1039/C4CC06165A},
}

@article{Kwak2020,
  author  = {Kwak, Jong Woo and Yilmaz, H. Birkan and Farsad, Nariman and Chae, Chan-Byoung and Goldsmith, Andrea J.},
  title   = {Two-Way Molecular Communications},
  journal = {IEEE Trans. Commun.},
  year    = {2020},
  month   = jun,
  volume  = {68},
  number  = {6},
  pages   = {3550--3563},
  doi     = {10.1109/TCOMM.2020.2977282},
}

@article{Baradari2025,
  author  = {Baradari, Elham and Akan, Ozgur B.},
  title   = {Molecular Communication-Based Intelligent Dopamine Rate Modulator for {Parkinson's} Disease Treatment},
  journal = {IEEE Trans. Nanobiosci.},
  year    = {2025},
  month   = apr,
  volume  = {24},
  number  = {2},
  pages   = {136--144},
  doi     = {10.1109/TNB.2024.3456031},
}

@article{Venton2003,
  author = {Venton, B. Jill and Zhang, Hui and Garris, Paul A. and Phillips, Paul E. M. and Sulzer, David and Wightman, R. Mark},
  title = {Real-time decoding of dopamine concentration changes in the caudate--putamen during tonic and phasic firing},
  journal = {J. Neurochem.},
  year = {2003},
  month = dec,
  volume = {87},
  number = {5},
  pages = {1284--1295},
  doi = {10.1046/j.1471-4159.2003.02109.x}
}

@article{Hershey2013,
  author = {Hershey, Neil D. and Kennedy, Robert T.},
  title = {In Vivo Calibration of Microdialysis Using Infusion of Stable-Isotope Labeled Neurotransmitters},
  journal = {ACS Chem. Neurosci.},
  year = {2013},
  month = may,
  volume = {4},
  number = {5},
  pages = {729--736},
  doi = {10.1021/cn300199m}
}

@article{Keithley2011,
  author = {Keithley, Richard B. and Takmakov, Pavel and Bucher, Elizabeth S. and Belle, Anna M. and Owesson-White, Catarina A. and Park, Jinwoo and Wightman, R. Mark},
  title = {Higher Sensitivity Dopamine Measurements with Faster-Scan Cyclic Voltammetry},
  journal = {Anal. Chem.},
  year = {2011},
  month = may,
  volume = {83},
  number = {9},
  pages = {3563--3571},
  doi = {10.1021/ac200143v}
}

@article{Letham2016,
  author = {Letham, Benjamin and Letham, Portia A. and Rudin, Cynthia and Browne, Edward P.},
  title = {Prediction uncertainty and optimal experimental design for learning dynamical systems},
  journal = {Chaos},
  year = {2016},
  month = jun,
  volume = {26},
  number = {6},
  eid = {063110},
  note = {{Art.} no. 063110},
  doi = {10.1063/1.4953795}
}

@misc{LomanBaker2025,
  author = {Loman, Torkel E. and Baker, Ruth E.},
  title = {Functional and parametric identifiability for universal differential equations applied to chemical reaction networks},
  year = {2025},
  month = oct,
  howpublished = {arXiv preprint arXiv:2510.14140},
  eprint = {2510.14140},
  archivePrefix = {arXiv},
  primaryClass = {math.DS},
  doi = {10.48550/arXiv.2510.14140}
}

@article{Cedersund2012,
  author = {Cedersund, Gunnar},
  title = {Conclusions via unique predictions obtained despite unidentifiability -- new definitions and a general method},
  journal = {FEBS J.},
  volume = {279},
  number = {18},
  pages = {3513--3527},
  year = {2012},
  month = sep,
  doi = {10.1111/j.1742-4658.2012.08725.x}
}

@misc{LomanBrowningBaker2026,
  author = {Loman, Torkel E. and Browning, Alexander P. and Baker, Ruth E.},
  title = {Structural functional identifiability and model discovery in differential equation models},
  howpublished = {arXiv preprint arXiv:2606.30289},
  year = {2026},
  month = jun,
  eprint = {2606.30289},
  archivePrefix = {arXiv},
  primaryClass = {math.ST},
  doi = {10.48550/arXiv.2606.30289}
}

@article{Zabini2023,
  author = {Zabini, Flavio and Masini, Barbara Mav{\'i}},
  title = {Performance Limits of Spatially Distributed Molecular Communications With Receiver Saturation},
  journal = {IEEE Trans. Mol., Biol. Multi-Scale Commun.},
  year = {2023},
  month = mar,
  volume = {9},
  number = {1},
  pages = {94--99},
  doi = {10.1109/TMBMC.2022.3225740}
}

@book{CoverThomas2006,
  author = {Cover, Thomas M. and Thomas, Joy A.},
  title = {Elements of Information Theory},
  edition = {2},
  publisher = {Wiley},
  year = {2006},
  isbn = {9780471241959},
  doi = {10.1002/047174882X}
}

@article{NiAkan2025,
  author = {Ni, Hongbin and Akan, Ozgur B.},
  title = {{ART-Rx}: A Proportional-Integral-Derivative ({PID}) Controlled Adaptive Real-Time Threshold Receiver for Molecular Communication},
  journal = {IEEE Trans. Mol., Biol. Multi-Scale Commun.},
  year = {2025},
  month = sep,
  volume = {11},
  number = {3},
  pages = {435--450},
  doi = {10.1109/TMBMC.2025.3581470}
}
\end{document}